\pdfoutput=1
\RequirePackage{ifpdf}
\ifpdf 
\documentclass[pdftex]{sigma}
\else
\documentclass{sigma}
\fi

\usepackage{bbm}

\numberwithin{equation}{section}

\newtheorem{Theorem}{Theorem}[section]
\newtheorem{Proposition}[Theorem]{Proposition}
\newtheorem{Corollary}[Theorem]{Corollary}

\begin{document}

\newcommand{\arXivNumber}{2105.12593}

\renewcommand{\thefootnote}{}

\renewcommand{\PaperNumber}{084}

\FirstPageHeading

\ShortArticleName{Exponential Formulas, Normal Ordering and the Weyl--Heisenberg Algebra}

\ArticleName{Exponential Formulas, Normal Ordering\\ and the Weyl--Heisenberg Algebra}

\Author{Stjepan MELJANAC~$^{\rm a}$ and Rina \v{S}TRAJN~$^{\rm b}$}

\AuthorNameForHeading{S.~Meljanac and R.~\v{S}trajn}

\Address{$^{\rm a)}$~Division of Theoretical Physics, Ruder Bo\v{s}kovi\'c Institute,\\
\hphantom{$^{\rm a)}$}~Bijeni\v{c}ka cesta 54, 10002 Zagreb, Croatia}
\EmailD{\href{mailto:meljanac@irb.hr}{meljanac@irb.hr}}

\Address{$^{\rm b)}$~Department of Electrical Engineering and Computing, University of Dubrovnik,\\
\hphantom{$^{\rm b)}$}~\'{C}ira Cari\'{c}a 4, 20000 Dubrovnik, Croatia}
\EmailD{\href{mailto:rina.strajn@unidu.hr}{rina.strajn@unidu.hr}}

\ArticleDates{Received May 27, 2021, in final form September 09, 2021; Published online September 15, 2021}

\Abstract{We consider a class of exponentials in the Weyl--Heisenberg algebra with exponents of type at most linear in coordinates and arbitrary functions of momenta. They are expressed in terms of normal ordering where coordinates stand to the left from momenta. Exponents appearing in normal ordered form satisfy differential equations with boundary conditions that could be solved perturbatively order by order. Two propositions are presented for the Weyl--Heisenberg algebra in 2 dimensions and their generalizations in higher dimensions. These results can be applied to arbitrary noncommutative spaces for construction of star products, coproducts of momenta and twist operators. They can also be related to the BCH formula.}


\Classification{16S32; 81R60}

\section{Introduction}

The first model of NC geometry was that of the Snyder spacetime~\cite{Snyder}, and NC spaces have since been used in theoretical physics in the
efforts to understand and model Planck scale phenomena~\cite{dfr1,dfr2,Luk2,Luk1, Snyder}.
Besides the Snyder spacetime, another widely studied model is the $\kappa$-Minkowski spacetime \cite{gacluknow,dasklukwor, Luk2,Luk1,majidruegg}, where the parameter $\kappa$ is usually interpreted as the Planck mass or the quantum gravity scale and the coordinates themselves close a Lie algebra. The $\kappa$-Poincar\'e quantum group \cite{ChPr, Luk2,Luk1,majid}, as a possible quantum symmetry of the $\kappa$-Minkowski spacetime, allows for the study of deformed relativistic spacetime symmetries and the corresponding dispersion relations \cite{asc-disp, MMMs2}. It is an example of a Hopf algebra, where the algebra sector is the same as that of the Poincar\'e algebra, but the coalgebra sector is deformed.

A powerful tool in the study of NC spaces is that of realizations in terms of the Weyl--Heisenberg algebra \cite{battisti,GGHMM, MMMs2,pikuticgupta}. Namely, the NC coordinates are expressed in terms of the commutative ones, which allows one to simplify the methods of calculation on the deformed spacetime. Exponential formulas \cite{battisti2,MMSSt1,MMSSt2,MMSSt3,MSS} are related to the deformed coproduct of momenta in NC spaces and also appear in the computations of star products, which are both needed for the definition of a field theory and the notion of differential calculus on a NC spacetime. In~\cite{MSS}, combinatorial recursions were found and formal differential equations were derived. Closely related problems were considered in~\cite{vis1,vis2,vis3}. Exponential formulae, which were used to obtain the coproducts and star products were also presented in \cite{BMMP, MMPP} and in a Lie deformed phase space from twists in the Hopf algebroid approach in \cite{JKM,JMsalg,LMMPW,LMW}.

Our motivation is to generalize the results found in \cite{BMMP, MMSSt1,MMSSt2,MMSSt3,MSS}. In Section~\ref{section2} we consider a class of exponentials in the Weyl--Heisenberg algebra with exponents of type at most linear in coordinates and arbitrary functions of momenta. They are expressed in terms of normal ordering where coordinates stand to the left of momenta. Exponents appearing in normal ordered form satisfy differential equations with boundary conditions that could be solved perturbatively order by order. Two propositions are presented for the Weyl--Heisenberg algebra in 2 dimensions. In Section~\ref{section3} we present these results for the Weyl--Heisenberg algebra in higher dimensions.

\section[Exponential formula for the Weyl--Heisenberg algebra in two dimensions]{Exponential formula for the Weyl--Heisenberg algebra\\ in two dimensions}\label{section2}

In this section we consider a class of exponentials in the Weyl--Heisenberg algebra in 2 dimensions (generated with one coordinate and the corresponding momentum) with exponents of type at most linear in the coordinate and arbitrary functions of momentum. They are expressed in terms of normal ordering in which powers of the coordinate stand to the left of the powers of momentum. Exponents appearing in normal ordered form satisfy differential equations with boundary conditions that could be solved perturbatively order by order. The Weyl algebra, quantum theory and normal ordering were presented in~\cite{mansour}, specially, physical aspects of normal ordering. Normal ordering in the Weyl algebra and relations in the extended Weyl algebra were also given~\cite{mansour}. Some functional analysis considerations concerning considered
operators, space, domain etc.\ can be found in~\cite{mansour}. We present two propositions that are generalizations of the results in \cite{BMMP,mansour, MSS}.

\begin{Proposition}\label{proposition1}
If $[ p,x ]=-{\rm i} $, then
\begin{itemize}\itemsep=0pt
\item[$i)$] ${\rm e}^{{\rm i}kx\varphi (p)} = {}{:} {\rm e}^{{\rm i}x\phi(k,p)}{:}$, $k\in \mathbb{R}$,
 where ${:}\, {:}$ denotes normal ordering, with the $x$'s left from the~$p$'s,
\item[$ii)$] $\phi(k,p)={\rm e}^{-{\rm i}kx\varphi(p)} p {\rm e}^{{\rm i}kx\varphi(p)}-p = J(k,p)-p$, i.e., $J(k,p)={\rm e}^{-{\rm i}kx\varphi(p)}p{\rm e}^{{\rm i}kx\varphi(p)}$,
\item[$iii)$] $J(k,p)$ is a unique solution of the partial differential equation
\begin{equation}\label{Prop1iii}
\frac{\partial J(k,p)}{\partial k}=\varphi(J(k,p)),
\end{equation}
with boundary condition $J(0,p)=p$.
\end{itemize}
\end{Proposition}

\begin{proof}
Expanding ${\rm e}^{{\rm i}kx\varphi(p)}=\sum_{n=0}^\infty \frac{1}{n!} ({\rm i}kx\varphi(p))^n$, and performing normal ordering of each term $(x\varphi(p))^n$, we are led to a general Ansatz
\begin{equation*}
{\rm e}^{{\rm i}kx\varphi(p)}=\sum_{n=0}^\infty \frac{({\rm i}x)^n}{n!}\phi_n(k,p),
\end{equation*}
for some $\phi_n(k,p)$,
\begin{equation*}
(\operatorname{ad}_p)^m {\rm e}^{{\rm i}kx\varphi(p)} = \sum_{n=m}^\infty\frac{({\rm i}x)^{n-m}}{(n-m)!} \phi_n(k,p),
\end{equation*}
where $\operatorname{ad}_p(x^n)=[p,x^n]=-{\rm i}nx^{n-1}$, and
\begin{equation*}
 \big( (\operatorname{ad}_p)^m {\rm e}^{{\rm i}kx\varphi(p)} \big)\big\vert_{x=0}= \phi_m(k,p).
\end{equation*}
Hence,
\begin{equation*}
{\rm e}^{{\rm i}kx\varphi(p)}=\sum_{n=0}^\infty \frac{({\rm i}x)^n}{n!} \big( (\operatorname{ad}_p)^n \big({\rm e}^{{\rm i}kx\varphi(p)}\big)\big)\big\vert_{x=0}.
\end{equation*}
Let us introduce
\begin{equation*}
J(k,p)={\rm e}^{-{\rm i}kx\varphi(p)}p{\rm e}^{{\rm i}kx\varphi(p)}= p+\phi(k,p).
\end{equation*}
Furthermore,
\begin{align*}
 \operatorname{ad}_p\big({\rm e}^{{\rm i}kx\varphi(p)}\big) & = \big[p,{\rm e}^{{\rm i}kx\varphi(p)}\big]= {\rm e}^{{\rm i}kx\varphi(p)} \big({\rm e}^{-{\rm i}kx\varphi(p)} p {\rm e}^{{\rm i}kx\varphi(p)} -p\big)\nonumber \\
& = {\rm e}^{{\rm i}kx\varphi(p)}(J(k,p)-p)= {\rm e}^{{\rm i}kx\varphi(p)} \phi(k,p).
\end{align*}
Similarly,
\begin{equation*}
(\operatorname{ad}_p)^n\big({\rm e}^{{\rm i}kx\varphi(p)}\big) ={\rm e}^{{\rm i}kx\varphi(p)} (\phi(k,p))^n.
\end{equation*}
It follows
\begin{equation*}
\phi_n(k,p)=(\operatorname{ad}_p)^n\big({\rm e}^{{\rm i}kx\varphi(p)}\big)\big\vert_{x=0}= (\phi(k,p))^n.
\end{equation*}
Hence
\begin{equation*}
{\rm e}^{{\rm i}kx\varphi(p)}= \sum_{n=0}^\infty \frac{({\rm i}x)^n}{n!} (\phi(k,p))^n= {:}{\rm e}^{{\rm i}x\phi(k,p)}{:}.
\end{equation*}
The differential equation for $J(k,p)$ is
\begin{align*}
\frac{\partial J(k,p)}{\partial k} &= {\rm e}^{-{\rm i}kx\varphi (p)} (-{\rm i}kx\varphi (p)p+{\rm i}px\varphi (p)) {\rm e}^{{\rm i}kx\varphi (p)}= {\rm e}^{-{\rm i}kx\varphi(p)} \varphi(p) {\rm e}^{{\rm i}kx\varphi(p)} \\
&= \varphi \big( {\rm e}^{-{\rm i}kx\varphi (p)}p{\rm e}^{{\rm i}kx\varphi (p)}\big)= \varphi(J(k,p)),
\end{align*}
with boundary condition $J(0,p)=p$.
\end{proof}

Let us introduce
\begin{equation*}
O=\operatorname{ad}_{-{\rm i}x\varphi(p)}= \varphi(p)\frac{\partial}{\partial p}.
\end{equation*}
Consequently,
\begin{equation*}
J(k,p)={\rm e}^{kO}(p), \qquad O(p)=\varphi(p),
\end{equation*}
and
\begin{equation*}
\phi(k,p) =J(k,p)-p=\left(\frac{{\rm e}^{kO}-1}{O}\right) (\varphi(p)), \qquad O(\varphi(p))= \varphi(p)\frac{\partial}{\partial p} \varphi(p).
\end{equation*}

For example if $\varphi(p) = p^l$ and ${\rm e}^{{\rm i}kx \varphi(p)}$ acts on ${\rm e}^{{\rm i}qx}$, $q \in \mathbb{R}$, then the solution of \eqref{Prop1iii} in iii) gives the same result as equation~(6,49) in~\cite{mansour}, changing ${\rm i}p \to {\rm d}/({\rm d} x)$, $(-{\rm i})^{l-1} k \to \lambda$, ${\rm i}q \to k$.

\begin{Proposition}\label{proposition2}
If $[p,x]=-{\rm i}$, then
\begin{equation*}
{\rm e}^{{\rm i}kx\varphi(p)+{\rm i}k\chi(p)}= {\rm e}^{{\rm i}kx\varphi(p)}{\rm e}^{{\rm i}h(k,p)}= {:}{\rm e}^{{\rm i}x\phi(k,p)}{:}\, {\rm e}^{{\rm i}h(k,p)}, \qquad k\in \mathbb{R},
\end{equation*}
where
\begin{equation}\label{eih}
{\rm e}^{{\rm i}h(k,p)}={\rm e}^{-{\rm i}kx\varphi(p)} {\rm e}^{{\rm i}kx\varphi(p)+{\rm i}k\chi(p)},
\end{equation}
and
\begin{equation*}
h(k,p)=\left(\frac{{\rm e}^{kO}-1}{O}\right)(\chi(p)).
\end{equation*}
$h(k,p)$ is the unique solution of the partial differential equation
\begin{equation*}
\frac{\partial h(k,p)}{\partial k}=\chi(J(k,p)),
\end{equation*}
with boundary condition $h(0,p)=0$.
\end{Proposition}

\begin{proof}
\begin{align*}
\frac{\partial h(k,p)}{\partial k}{\rm e}^{{\rm i}h(k,p)} &= -{\rm i}\frac{\partial}{\partial k}\left( {\rm e}^{-{\rm i}kx\varphi (p)}{\rm e}^{{\rm i}kx\varphi (p)+{\rm i}k\chi(p)}\right) \\
&= {\rm e}^{-{\rm i}kx\varphi (p)}\left( -x\varphi (p)+x\varphi (p)+\chi(p)\right) {\rm e}^{{\rm i}kx\varphi (p)+{\rm i}x\chi(p)} \\
& = {\rm e}^{-{\rm i}kx\varphi(p)}\chi(p) {\rm e}^{{\rm i}kx\varphi(p) +{\rm i}k\chi(p)} \\
&= \chi\big( {\rm e}^{-{\rm i}kx\varphi (p)} p{\rm e}^{{\rm i}kx\varphi(p)}\big) {\rm e}^{-{\rm i}kx\varphi(p)} {\rm e}^{{\rm i}kx\varphi(p)+{\rm i}k\chi(p)}= \chi(J(k,p)){\rm e}^{{\rm i}h(k,p)}.
\end{align*}
The partial differential equation for $h(k,p)$ follows from the above equation. The unique solution for $h(k,p)$ is
\begin{equation*}
h(k,p)=\left(\frac{{\rm e}^{kO}-1}{O}\right)(\chi(p))=k\sum_{n=0}^\infty \frac{1}{(n+1)!}k^n O^n(\chi(p)),
\end{equation*}
in agreement with the BCH formula
\begin{gather*}
{\rm e}^A {\rm e}^B=\exp \left( A+B+\frac{1}{2}[A,B] +\frac{1}{12}([A,[A,B]]+[B,[B,A]])+\cdots \right)
\end{gather*} and $A=-{\rm i}kx\varphi(p)$, $B={\rm i}kx\varphi(p)+ {\rm i}k\chi(p)$, see the r.h.s.\ of~\eqref{eih}.
\end{proof}

\section[Exponential formula for the Weyl--Heisenberg algebra in higher dimensions]{Exponential formula for the Weyl--Heisenberg algebra\\ in higher dimensions}\label{section3}

Propositions~\ref{proposition1} and~\ref{proposition2} can be generalized to the Weyl--Heisenberg algebra in higher dimensions, which leads to two theorems. Here the Weyl--Heisenberg algebra is generated with coordinates~$x_\mu$ and momenta $p_\nu$ with Minkowski metric $\eta_{\mu \nu} = \operatorname{diag}(-1,1,\dots ,1)$ and $\mu, \nu = 0,1,\dots ,n-1$. Summation over repeated indices is assumed. Generalization to arbitrary signature is straightforward.

\begin{Theorem}\label{theorem1}
If $[x_\mu,x_\nu]=0$, $[p_\mu,p_\nu]=0$ and $[p_\mu,x_\nu]=-{\rm i}\eta_{\mu\nu}$, then
\begin{itemize}\itemsep=0pt
\item[$i)$] ${\rm e}^{{\rm i}k_\beta x_\alpha \varphi_{\alpha\beta}(p)+{\rm i}k_\alpha \chi_\alpha (p)}= {:}{\rm e}^{{\rm i}x_\alpha \phi_\alpha(k,p)}{:}\, {\rm e}^{{\rm i}h(k,p)}$, $ k=(k_\mu) \in M_n$, where $M_n$ denotes Minkowski space.
\item[$ii)$] $\phi_\mu (k,p)=\big( {\rm e}^{O}\big)(p_\mu)-p_\mu= \big( \frac{{\rm e}^{O}-1}{O}\big) (k_\beta \varphi_{\mu \beta}(p))= J_\mu(k,p)-p_\mu$,
where
$O=k_\alpha O_\alpha$ and $O_\alpha =\operatorname{ad}_{(-{\rm i}x_\beta \varphi_{\beta\alpha}(p))}$.
\item[$iii)$] $J_\mu(k,p)={\rm e}^{-{\rm i}k_\beta x_\alpha \varphi_{\alpha\beta}(p)}p_\mu {\rm e}^{{\rm i}k_\beta x_\alpha\varphi_{\alpha\beta(p)}}= \big( {\rm e}^{k_\alpha O_\alpha}\big)(p_\mu)$,
$J_\mu(k,p)$ is the unique solution of the partial differential equation
\begin{equation*}
k_\beta \frac{\partial J_\mu(k,p)}{\partial k_\beta}=k_\beta \varphi_{\mu\beta}(J(k,p)),\qquad J_\mu(0,p)=p_\mu.
\end{equation*}
\item[$iv)$] ${\rm e}^{{\rm i}h(k,p)}={\rm e}^{-{\rm i}k_\beta x_\alpha \varphi_{\alpha \beta}(p)} {\rm e}^{{\rm i}k_\beta x_\alpha \varphi_{\alpha\beta}(p)+{\rm i}k_\alpha \chi_\alpha (p)}$,
$h(k,p)=\big(\frac{{\rm e}^{k_\alpha O_\alpha}-1}{k_\alpha O_\alpha}\big)(k_\beta \chi_\beta(p))$.
$h(k,p)$ is the unique solution of the partial differential equation
\begin{equation*}
k_\beta \frac{\partial h(k,p)}{\partial k_\beta}=k_\beta \chi_\beta(J(k,p)), \qquad h(0,p)=0.
\end{equation*}
\end{itemize}
\end{Theorem}
The proof of Theorem~\ref{theorem1} follows from Propositions~\ref{proposition1} and~\ref{proposition2} in a straightforward way.

Theorem~\ref{theorem1} has been used to construct the Drinfeld twist corresponding to a
linear realization of the $\mathfrak{gl}(n)$ algebra~\cite{romp}.

For example $J_\mu(k,p)$ in Theorem~\ref{theorem1}(iii), in the 3rd order in the ad expansion is given by
\begin{gather}
J_\mu(k,p) = p_\mu +(k_\alpha O_\alpha)(p_\mu) +\frac{1}{2}(k_\alpha O_\alpha)^2(p_\mu)+\frac{1}{6}(k_\alpha O_\alpha)^3(p_\mu)\nonumber \\
\hphantom{J_\mu(k,p)}{} = p_\mu +k_\alpha \varphi_{\mu\alpha}(p) +\frac{1}{2}k_{\alpha'}\varphi_{\beta'\alpha'}(p) \frac{\partial}{\partial p_{\beta'}} (k_\alpha \varphi_{\mu\alpha}(p)) \nonumber \\
\hphantom{J_\mu(k,p)=}{} +\frac{1}{6} k_{\alpha''}\varphi_{\beta''\alpha''}(p) \frac{\partial}{\partial p_{\beta''}}\left( k_{\alpha'} \varphi_{\beta'\alpha'} (p)\frac{\partial}{\partial p_{\beta'}}(k_\alpha \varphi_{\mu\alpha}(p))\right).\label{J3reda}
\end{gather}
If we perform a transformation $k_\mu \to \lambda k_\mu$, $\lambda \in \mathbb{R}$, Theorem~\ref{theorem1} becomes

\begin{Theorem}\label{theorem2}\quad
\begin{itemize}\itemsep=0pt
\item[$i)$] ${\rm e}^{{\rm i}\lambda k_\beta x_\alpha \varphi_{\alpha\beta}(p)+ {\rm i}\lambda k_\alpha \chi_\alpha (p)} = {:} {\rm e}^{{\rm i}x_\alpha \phi_\alpha (\lambda k,p)}{:}\, {\rm e}^{h(\lambda k,p)}$.
\item[$ii)$] $\phi_\mu (\lambda k,p)=\big({\rm e}^{\lambda O}\big)(p_\mu)-p_\mu = \big( \frac{{\rm e}^{\lambda O}-1}{O}\big) (k_\beta \varphi_{\mu\beta}(p))= J_\mu(\lambda k,p) -p_\mu$, $O=k_\alpha O_\alpha$.
\item[$iii)$] $J_\mu(\lambda k,p)={\rm e}^{-{\rm i}\lambda k_\beta x_\alpha \varphi_{\alpha\beta}(p)}p_\mu {\rm e}^{{\rm i}\lambda k_\beta x_\alpha \varphi_{\alpha\beta}(p)}= \big({\rm e}^{\lambda O}\big)(p_\mu)$.
$J_\mu(\lambda k,p)$ is the unique solution of the partial differential equation
\begin{equation*} 
\frac{\partial J_\mu (\lambda k,p)}{\partial \lambda}= k_\beta \varphi_{\mu \beta}(J(\lambda k,p)), \qquad J_\mu (0,p)=p_\mu.
\end{equation*}
\item[$iv)$] ${\rm e}^{{\rm i}h(\lambda k,p)}={\rm e}^{-{\rm i}\lambda k_\beta x_\alpha \varphi_{\alpha\beta}(p)}{\rm e}^{{\rm i}\lambda k_\beta x_\alpha \varphi_{\alpha\beta}(p)+ {\rm i}\lambda k_\alpha \chi_\alpha (p)}$,
$h(\lambda k,p)=\big( \frac{{\rm e}^{\lambda O}-1}{O}\big)(k_\beta \chi_\beta (p))$.
$h(\lambda k,p)$ is the unique solution of the partial differential equation
\begin{equation*}
\frac{\partial h(\lambda k,p)}{\partial \lambda} =k_\beta \chi_\beta (J(\lambda k,p)),
\end{equation*}
with boundary condition $h(0,p)=0$.
\end{itemize}
\end{Theorem}
Note that
\begin{equation*}
k_\beta \frac{\partial}{\partial k_\beta} \to \lambda k_\beta \frac{\partial}{\partial (\lambda k_\beta)}=\lambda \frac{\partial}{\partial \lambda}.
\end{equation*}
\begin{Corollary}
If $\varphi_{\mu\nu}(p)$ is linear in $p$, then $\phi_\mu(k,p)$ is also linear in $p$ and vice versa,
\begin{equation*}
{\rm e}^{{\rm i}\lambda x_\alpha A_{\alpha\beta}p_\beta}={:}{\rm e}^{{\rm i}x_\alpha ({\rm e}^{\lambda A}- \mathbbm{1} )_{\alpha\beta}p_\beta}{:}, \qquad \lambda \in \mathbb{R},
\end{equation*}
where $A$ is a matrix of scalar constants.
\end{Corollary}
\begin{proof}
Let $k_\alpha \varphi_{\mu\alpha}(p)=A_{\mu\alpha}p_\alpha$, then according to Theorem~\ref{theorem1} $J_\mu(\lambda k,p)=\big({\rm e}^{\lambda O}\big)(p_\mu)= \big({\rm e}^{\lambda A}\big)_{\mu\alpha} p_\alpha$ and $\phi_\alpha (\lambda k,p)=\big( {\rm e}^{\lambda A} -\mathbbm{1}\big)_{\alpha\beta}p_\beta$.

If $\phi_\mu (k,p)$ and $J_\mu(k,p)$ are not linear in $p$, then $\varphi_{\mu\nu}(p)$ is also not linear in $p$. From the expressions for $J_\mu(k,p)$ and $h(k,p)$ in Theorem~\ref{theorem1}(iii)(iv), we find
\begin{equation*}
\varphi_{\mu\nu}(p) =\left. \frac{\partial J_\mu(k,p)}{\partial k_\nu}\right\vert_{k=0}, \qquad \chi_\mu (p)= \left. \frac{\partial h(k,p)}{\partial k_\mu}\right\vert_{k=0}.
\end{equation*}
Or from Theorem~\ref{theorem2}(iii)(iv)
\begin{equation*}
k_\alpha \varphi_{\mu\alpha}(p)= \left.\frac{J_\mu(\lambda k,p)}{\partial \lambda}\right\vert_{\lambda=0}, \qquad k_\alpha\chi_\alpha(p) =\left.\frac{\partial h(\lambda k,p)}{\partial \lambda}\right\vert_{\lambda=0}.
\end{equation*}
On \looseness=1 the other hand, from $\varphi_{\mu\nu}(p)$, we can construct $J_\mu(k,p)$ and $\phi_\mu(k,p)$ perturbatively order by order using Theorem~\ref{theorem1}(iii) or Theorem~\ref{theorem2}(iii), see equation~\eqref{J3reda}. Analogously, from~$\chi_\mu(p)$ we can construct $h(k,p)$ perturbatively order by order using Theorem~\ref{theorem1}(iv) or Theo\-rem~\ref{theorem2}(iv).
\end{proof}

Let us define the $\triangleright$ action by
\begin{equation*}
x_\mu \triangleright f(x)= x_\mu f(x), \qquad p_\mu \triangleright f(x)= -{\rm i}\frac{\partial f(x)}{\partial x_\mu}.
\end{equation*}
These relations imply $p_\mu \triangleright {\rm e}^{{\rm i}qx}=q_\mu {\rm e}^{{\rm i}qx}$, whence $f(p) \triangleright {\rm e}^{{\rm i}qx}=f(q){\rm e}^{{\rm i}qx}$. Then we have,
\begin{Corollary}
\begin{equation*}
{\rm e}^{{\rm i}k_\beta x_\alpha \varphi_{\alpha\beta}(p)+{\rm i}k_\alpha \chi_\alpha (p)} \triangleright {\rm e}^{{\rm i}q_\alpha x_\alpha}= {\rm e}^{{\rm i}x_\alpha J_\alpha (k,q)+{\rm i}h(k,q)}.
\end{equation*}
\end{Corollary}
\begin{Corollary}
Let $\lambda \in \mathbb{R}$ and let
\begin{align*}
{\rm e}^{{\rm i}\lambda x_\alpha F_\alpha^{(1)}(p)} = {:}{\rm e}^{{\rm i}x_\alpha (J_\alpha^{(1)} (\lambda ,p)-p_\alpha)}{:},\\
{\rm e}^{{\rm i}\lambda x_\alpha F_\alpha^{(2)}(p)} = {:}{\rm e}^{{\rm i}x_\alpha (J_\alpha^{(2)} (\lambda ,p)-p_\alpha)}{:},
\end{align*}
where
\begin{gather*}
J_\mu^{(1)}(\lambda,p)=\big( {\rm e}^{\lambda O^{(1)}}\big) (p_\mu), \qquad O^{(1)}= \operatorname{ad}_{-{\rm i}x_\alpha F_\alpha^{(1)}(p)},\\
J_\mu^{(2)}(\lambda,p)=\big( {\rm e}^{\lambda O^{(2)}}\big) (p_\mu), \qquad O^{(2)}= \operatorname{ad}_{-{\rm i}x_\alpha F_\alpha^{(2)}(p)}.
\end{gather*}
Then,
\begin{gather*}
 {\rm e}^{{\rm i}\lambda x_\alpha F_\alpha^{(1)}(p)} {\rm e}^{{\rm i}\lambda x_\beta F_\beta^{(2)}(p)}= {\rm e}^{{\rm i}\lambda x_\gamma F_\gamma^{(3)}(p)}={:}{\rm e}^{{\rm i}x_\gamma (J_\gamma^{(3)}(\lambda,p)-p_\gamma)}{:},\\
 J_\mu ^{(3)}(\lambda,p)=J_\mu ^{(1)}\big(\lambda,J^{(2)}(\lambda,p)\big).
\end{gather*}
\end{Corollary}

\begin{proof}
\begin{equation*}
{\rm e}^{{\rm i}\lambda x_\alpha F_\alpha^{(1)}(p)}\,{\rm e}^{{\rm i}\lambda x_\beta F_\beta^{(2)}(p)}\triangleright {\rm e}^{{\rm i}qx}={\rm e}^{{\rm i}\lambda x_\alpha F_\alpha^{(1)}(p)} \triangleright {\rm e}^{{\rm i}x_\alpha J_\alpha^{(2)}(\lambda,q)} = {\rm e}^{{\rm i}x_\alpha J_\alpha^{(1)}(\lambda,J^{(2)}(\lambda,q))},
\end{equation*}
from which it follows
\begin{gather*}
J_\mu ^{(3)}(\lambda,q)=J_\mu ^{(1)}\big(\lambda,J^{(2)}(\lambda,q)\big), \qquad \forall\, q \in \mathbb{R},\\
J_\mu ^{(3)}(\lambda,p)= J_\mu ^{(1)}\big(\lambda,J^{(2)}(\lambda,p)\big).
\end{gather*}
From $J_\mu ^{(3)}(\lambda,p)= \big( {\rm e}^{\lambda O^{(3)}}\big) (p_\mu)$ and $O^{(3)}= \operatorname{ad}_{-ix_\alpha F_\alpha^{(3)}(p)}$ we can construct $F_\mu ^{(3)}(p)= \frac{\partial J^{(3)}_\mu(\lambda,p)}{\partial \lambda}\big\vert_{\lambda=0}$. This is an alternative construction to the one obtained from the BCH expansion.
\end{proof}

The results from Section~\ref{section3} can be applied to arbitrary noncommutative spaces for the construction of star products, coproducts of momenta and twist operators \cite{BMMP, MMPP}, specially for the Weyl realizations of Lie algebra type spaces. They can also be related to the BCH formula.

\pdfbookmark[1]{References}{ref}
\LastPageEnding

\end{document}